\RequirePackage{fix-cm}
\documentclass[twocolumn,smallextended]{svjour3}       
\smartqed  
\usepackage{graphicx}
\usepackage{amsmath} 
\usepackage{amsfonts}
\usepackage{amssymb}

\usepackage{algorithm}
\usepackage{algpseudocode}
\usepackage{subfigure}
\usepackage{caption}
\usepackage{breqn}
\usepackage[authoryear]{natbib}

\begin{document}
	
	\title{Adaptive Target Tracking with a Mixed Team of Static and Mobile Guards: Deployment and Activation Strategies
	}
	
	
	\author{Guillermo J. Laguna \and Sourabh Bhattacharya   
	}
	
	
	\institute{Guillermo J. Laguna \at
		Department of Mechanical Engineering \\
		Iowa State University\\
		Ames, IA USA
		Tel.: +521-462-1905156\\
		\email{gjlaguna@iastate.edu}           
		\and
		Sourabh Bhattacharya \at
		Department of Mechanical Engineering \\
		Iowa State University\\
		Ames, IA USA
		\email{sbhattac@iastate.edu}
	}

	\date{Received: date / Accepted: date}

	\maketitle
	
		\begin{abstract}
			This work explores a variation of the art gallery problem in which a team of static and mobile guards track a mobile intruder with unknown maximum speed. First, we present an algorithm to identify {\it candidate vertices} in a polygon at which either static guards can be placed or they can serve as endpoints of the trajectory of mobile guards. Based on the triangulation of the polygon and the deployment of the guards we propose an allocation technique for the guards, such that each one of them is assigned to guard a subregion of the environment when the intruder is inside it. The allocation strategy leads to a classification of the guards based on their task and coordination requirements. Next, we present an activation strategy for the static guards that is adaptive to the instantaneous speed of the intruder. The deployment and the activation technique guarantee that a variable speed intruder is successfully tracked. Simulation results are presented to validate the efficacy of the proposed techniques.
		\end{abstract}
		\begin{keywords} {Target tracking, Mobile coverage, Sensor networks.}
		\end{keywords}
	\section{Introduction}
	\label{sec:intro}
	
	In the last two decades, there has been a widespread deployment of multi-robot systems as sensor platforms for gaining situational awareness in military \cite{Simon:2004} as well as civilian \cite{Werner:2006} applications. Robots deployed in such scenarios are capable of autonomously making decisions and taking actions with minimal human intervention. Prior information regarding the exogenous processes that effect the performance of such robotic systems is an essential component of their decision-making process \cite{theodoridis2012toward}. However, in reality, the system seldom has a complete knowledge about the uncertainty in the environment, especially when adversaries are involved. Due to the sensitive nature of applications, a system designer either uses a Bayesian approach \cite{ho1964bayesian} to maximize the expected reward or a Game-theoretic approach \cite{bacsar2008h} to guarantee a minimum reward in such scenarios. In this work, we explore a different direction, i.e., an {\it adaptive} approach to address the challenge that arises due to incomplete information. To be specific, we consider the problem of deploying a mixed team of static and mobile sensors for persistent observation of a mobile intruder. The team of sensors lack complete information regarding the motion model of the intruder. The deployment technique for the sensors proposed in this paper is in part motivated from guard deployment techniques prevalent in art-gallery problems \cite{O'Rourke:1983}.
	
	An important challenge for a network designer is to balance the trade-off between energy consumption of the network \cite{Bhatti:2009} and its overall coverage \cite{Howard:2002}. Energy efficiency is an integral component in the design of sensor networks to increase their average lifetime. A frequently used approach to extend the average lifetime of a sensor network is to schedule activation cycles of the individual sensors \cite{Ying:2006}. For example, \cite{Tian:2002} proposes scheduling algorithms that randomly assign sensor nodes to multiple working subsets, which alternately perform the sensing tasks for monitoring a sensor field. In \cite{Deng:2005}, authors propose a scheduling algorithm for clustered sensor networks that deactivates nodes based on their relative distance to the cluster head to improve energy saving. In addition to energy efficiency, scheduling schemes often need to fulfill a secondary objective, for example, preserving the transmission range \cite{Chen:2007} or the sensing coverage \cite{Kumar:2004} of the network. For a static camera network, reducing the volume of data collected is an important secondary objective since it lowers the overall latency of the network \cite{aghajan2009multi} by reducing the communication and processing overheads. A mobile camera network can alleviate this problem of {\it data deluge} \cite{baraniuk2011more} by reducing the number of sensors that need to be active. The solution to the persistent tracking problem explored in this paper is a collaborative scheme between a team of mobile sensors aided by a team of static sensors.
	
	In this work, we leverage results from art-gallery problems \cite{O'Rourke:1983}. The art gallery problem originates from a real-world problem of guarding an art gallery with a team of guards who together can observe the whole gallery. For an infinite speed intruder, the problem of persistently tracking an intruder with a set of sensors is equivalent to the problem of distributing sensors in the environment such that every region inside the environment is within the sensing range of at least one sensor (it is covered) at all times, which is an instance of the area coverage problem \cite{Howard:2002}. Results from the art-gallery problem can be used to deploy static guards \cite{O'Rourke:1983,Hoffmann:1990} in a polygon to cover it. For an intruder with finite speed, fewer guards can be deployed to track the intruder. In \cite{Laguna:2016,Laguna:2017}, we proposed a tracking strategy that deploys at most $\lfloor\frac{n}{4}\rfloor$ diagonal guards in a polygon with $n$ sides, and derived the maximum speed of the intruder for which the tracking strategy works for a fixed speed of the guards. All the aforementioned works assume that the maximum speed of the intruder is known \textit{a priori} which is a severe limitation since knowledge about the adversary is limited in surveillance applications. {\bf In contrast, this paper deals with the scenario when the guards do not have any \textit{a priori} knowledge about the maximum speed of the intruder.} 
	
	Keeping a mobile target within its sensing range in an environment containing obstacles gives rise to a path planning problem for the mobile observers/sensors. In visibility-based target tracking, the observers/sensors are equipped with vision sensors to track the mobile targets. A detailed review regarding prior work in tracking a single target with one mobile observer is provided in \cite{Bhattacharya:2010,Bhattacharya:20112}. For multiple mobile observers, centralized \cite{Mehmetcik:2013} as well as decentralized techniques \cite{Luke:2005} have been proposed to track a single target \cite{Hausman:2015,Williams:2015} or a team of targets \cite{Murrieta:2002,zou2015visibility,Jung:2002}. In contradistinction to the previous works that assume a free observer with no constraints in motion, this paper assumes that the observer is either static or restricted to move along a line segment inside the polygon. For the sake of clarity, we would like to mention here the distinction between our problem and that of {\it target search} which has received significant attention in the past \cite{LaValle:1997,Suzuki:1992}. In target search, the objective is to deploy a team of mobile robots to locate a target inside an environment, whereas we assume that the target is initially visible to the observer in our current work. 

	To be specific, the contributions of this work are as follows: (i) We present a strategy to deploy a team of static and diagonal guards to track an intruder in a polygonal environment without \textit{a priori} knowledge about its maximum speed. (ii) We present a scheme to activate the guards based on the instantaneous speed of the intruder. Therefore, the proposed tracking strategy is {\it adaptive} in the sense that it activates/deactivates the guards based on the instantaneous speed of the intruder. (iii) We show that the proposed tracking strategy coincides with the strategy to cover a polygon with mobile guards for low speed regimes of the intruder \cite{O'Rourke:1983}. In contrast, for high speed regimes of the intruder, the tracking strategy coincides with the strategy to cover a polygonal environment with static guards.
	
	The paper is organized as follows. In Section \ref{sec:problem}, we present the problem formulation. In Section \ref{sec:deployment}, we present the deployment strategy for the guards. In Section \ref{sec:crit}, the guards are classified depending on the trajectories selected for them in their deployment and a motion strategy for tracking is presented. In Section \ref{sec:activ}, we present the activation strategy for the guards. In Section \ref{sec:examp}, we present the results of the proposed technique in different environments. Section \ref{sec:conclusion} presents the conclusions and future work.

	\section{Motivation and Problem Formulation}
	\label{sec:problem}
	
	Consider a mobile intruder $I$ inside a simple $n$-gon ($n$ edges/vertices) $P$ with no holes. Let $v_e=v_e(t)\in[0,v_e^{max}]$ denote the speed of $I$. Let $p_I=p_I(t)$ denote the location of $I$ at time $t$. A set of guards $S_g'$ are deployed in $P$. We assume that each guard is equipped with an omnidirectional vision sensor with infinite range. As a result, the guard can see everything that lies inside its visibility polygon \cite{Berg:2008}. In this work, we assume that the guards do not have \textit{a priori} knowledge about the maximum speed of the intruder and can only measure its instantaneous speed.
	

	Next, we define some mathematical objects related to a polygon. A {\it triangulation} of a polygon $P$ is defined as a (not necessarily unique) partition of $P$ into a set of disjoint triangles such that the vertices of the triangles are vertices of the polygon, and there is no intersection between any pair of edges of the triangles. An edge of a triangle in the triangulation is called a \textit{diagonal} \cite{O'Rourke:1983}. The triangulation of a polygon $P$ can be represented as a planar graph $G=G(P)$ called a \textit{triangulation graph}. $V(G)$ (vertex set of $G$) corresponds to the vertices of $P$, $E(G)$ (edge set of $G$) corresponds to the diagonals of the triangulation of $P$, and $T(G)$ (triangle set of $G$) corresponds to the faces\footnote{In a triangulation graph all the faces are triangles.} of $G$. Clearly, there is a bijection between the set of vertices of $P$ and $V(G)$. Additionally, there is a bijection between the set of diagonals of the triangulation of $P$ and $E(G)$. As a result, we do not make a distinction between the following: (i) vertices of $P$ and the vertices in $V(G)$ (ii) diagonals of the triangulation of $P$ and the edges in $E(G)$ (iii) triangles of the triangulation of $P$ and the triangles in $T(G)$. Let $D$ be the dual graph of $G$. Each vertex in $V(D)$ corresponds to a triangle in $T(G)$, and there is an edge in $E(D)$ only between vertices that correspond to triangles in $T(G)$ that share an edge. For a simple polygon $P$, the dual graph of any triangulation of $P$ is a tree \cite{Berg:2008}.
	
	We define $S_g$ and $S_g^v$ as the sets of mobile and vertex guards, respectively. $S_g$ and $S_g^v$ are disjoint sets ($S_g \cap S_g^v = \emptyset$) such that $S_g \cup S_g^v = S_g'$. Each $g_i \in S_g^v$ is located at a vertex $v_i \in V(G)$ and is always static. Moreover, each $g_i \in S_g^v$ can either be in an \textit{active mode} (consuming power with its camera switched on) or in \textit{sleep mode} (negligible power consumption with its camera switched off). A vertex guard is said to be activated (deactivated) when it goes from the sleep mode to an active mode (active mode to a sleep mode). Next, we describe mobile guards. All guards in $S_g$ are {\it diagonal guards} \cite{O'Rourke:1983}, i.e., each $g_i \in S_g$ is a mobile guard associated to a distinct diagonal $h_i$ of the triangulation of $P$. To be more specific the trajectory of each $g_i \in S_g$ is defined by a diagonal $h_i \in E(G)$, so $g_i$ is only allowed to move along $h_i$. The endpoints of $h_i$ are denoted by $v_j(i)$ with $j \in \{1,2\}$. Diagonal guards can cover\footnote{In this work we say that a region is \textit{covered} by a guard if the region is inside the visibility polygon of the guard.} different regions of the environment depending on the location along their diagonals. Although diagonal guards are assumed to be mobile, they can choose to be static (as vertex guards) in case the regions that are assigned to them can be covered from one endpoint of their diagonal. Thus, they can stay motionless in certain situations. However, in contrast with vertex guards, they are always in active mode. A diagonal guard is assumed to have a maximum speed $0<\bar{v}_g<\infty$.  
	
	Let $r=v_e/\bar{v}_g$ denote the {\it speed ratio}. Without any \textit{a priori} knowledge about the maximum value of $v_e$, the deployment strategy should be capable of handling all possible values of $r$. For $r=\infty$, either $v_e=\infty$ or the guards are static ($\bar{v}_g=0$). In either case, the guards need to cover every region inside $P$. $\lfloor n/3 \rfloor$ activate static guards are sufficient to cover the polygon since \cite{O'Rourke:1987} shows that there is a set $S_c \subset V(G)$ of at most $\lfloor n/3 \rfloor$ vertices that can dominate\footnote{A triangulation graph $G$ is dominated by a set of vertices $S_c$ if at least one vertex of each triangle in $T(G)$ is an element of $S_c$.} $G$, which trivially implies that they can cover $P$. However, fewer than $\lfloor n/3 \rfloor$ activated guards might suffice for finite $r$.
	
	A sufficient condition for tracking is to ensure that there is always at least one guard covering the triangle\footnote{A triangle is said to be covered by a guard $g_i$, if $g_i$ is located at the boundary of the triangle.} where the intruder is located. Therefore, we do not require all the faces of the triangulation to be covered at the same time, only a subset of them. By allowing some guards of $S_g'$ to be mobile, it is possible to cover such subset of triangles with fewer active guards. Our objective is to determine a deployment strategy for a team of at most $\lfloor n/3 \rfloor$ guards such that at most $\lfloor n/4 \rfloor$ of them are diagonal guards and the rest are vertex guards. Moreover, we want to obtain an activation/deactivation strategy for the vertex guards which ensures that for any speed $v_e$, the intruder will always be observed by activating a sufficient number of vertex guards to aid the diagonal guards. The deployment strategy guarantees that for small speed regimes of the intruder only the diagonal guards are activated, while for high speed regimes all the guards (including the vertex guards) are activated\footnote{For high speed regimes, when all the vertex guards are activated, all the diagonal guards become static and they locate themselves at one of the vertices of their diagonals} to ensure coverage to track the intruder.
	
	Figure \ref{fig:motiv_examp} illustrates an example to motivate our problem. Figure \ref{fig:motiv_examp} (a) shows a triangulated polygon. Circles shaded in red represent guards ($g_1$ and $g_2$), and the red segments represent the diagonals allocated to the guards ($h_1$ and $h_2$, respectively). The circle shaded in black denotes the intruder $I$. The circle shaded in yellow is a vertex guard $g_3$ in sleep mode, which is located at endpoint $v_1(2)$ of $h_2$. $g_1$ and $g_2$ are sufficient to track the intruder at small values of $v_e$. However, beyond a critical value of $v_e$, $g_3$ needs to be activated to track the intruder since the maximum speed of $g_1$ and $g_2$ may not be sufficient to guarantee persistent tracking. This is illustrated in Figure \ref{fig:motiv_examp} (b) where we assume that a sufficiently fast intruder can ``hide'' inside triangle $T_1$. In that case, the inactive vertex guard $g_3$ is activated, from its location it covers triangles $T_1$, $T_2$ and $T_3$.
	

	
	
	\begin{figure} [h!]
		
		\begin{center}
		\includegraphics[width=1\linewidth,height=0.48\linewidth]{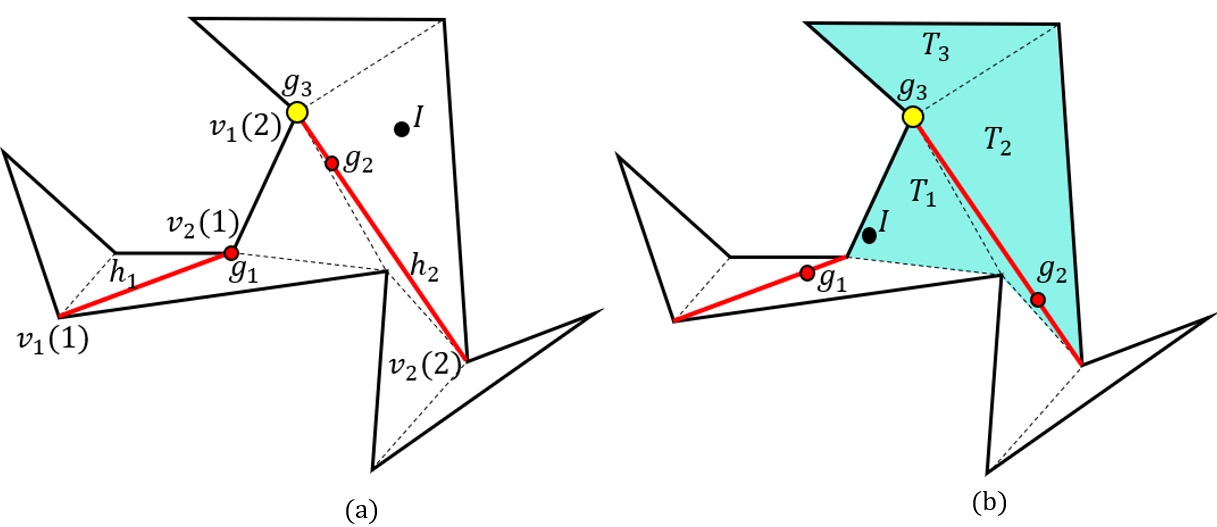}
		\end{center}
		
		\caption{(a) Two diagonal guards $g_1$ and $g_2$ are tracking $I$, $g_3$ is a deactivated vertex guard. (b) Activation of $g_3$ to aid $g_1$ and $g_2$.}
		\label{fig:motiv_examp}
		
	\end{figure}
	
	
	\section{Deployment of Guards}
	\label{sec:deployment}
	
	In this section, we propose a deployment strategy for the guards for tracking the intruder based on its instantaneous speed. The outline of the deployment is as follows. A set $S_c$ of at most $\lfloor n/3 \rfloor$ vertices of $G$ that dominate $G$ must be found, such that each vertex in $S_c$ is associated to a guard in $S_g'$. Since $S_g'=S_g \cup S_g^v$ and $|S_g| \leq \lfloor n/4 \rfloor$, where $|S_g|$ is the cardinality of $S_g$, then at most $\lfloor n/4 \rfloor$ vertices in $S_c$ correspond to diagonal guards and the remaining vertices in $S_c$ correspond to vertex guards. For the case of a vertex guard $g_i \in S_g^v$, the vertex in $S_c$ associated to $g_i$ corresponds to the location of $g_i$ in the triangulation of $P$. In contrast, for the case of a diagonal guard $g_i \in S_g$, the vertex in $S_c$ associated to $g_i$ does not correspond to the location of $g_i$ since it is a diagonal guard. Hence, for a diagonal guard $g_i$, the associated vertex in $S_c$ corresponds to an endpoint of $h_i$. Additionally, we require that for low speed regimes of the intruder, the diagonal guards must be able to guarantee persistent tracking. This implies that the set of diagonals $S_h=\{h_i \in E(G): g_i \in S_g  \}$ must be a set of \textit{dominating diagonals}\footnote{A triangulation graph $G$ is said to be \textit{dominated by a set of diagonals} $S_h$ if at least one vertex of each triangle in $T(G)$ is an endpoint of a diagonal in $S_h$.}. Therefore, $S_c$ is a set of at most $\lfloor n/3 \rfloor$ vertices of $G$ that dominates $G$, such that there is a dominating set $S_h$ of at most $\lfloor n/4 \rfloor$ diagonals, where an endpoint of each $h_i \in S_h$ is a different vertex in $S_c$.
	
	We call $S_c$ as \textit{the set of candidate vertices}. Since the vertices in $S_c$ dominate $G$ and each guard in $S_g'$ corresponds to a different vertex in $S_c$. The definition of $S_c$ guarantees that for infinite speed guards, locating every guard at its corresponding vertex in $S_c$ and activating every vertex guard suffices to cover all triangles in $T(G)$, which implies that persistent tracking is guaranteed with at most $\lfloor n/3 \rfloor$ guards regardless of the speed of the intruder, while for speeds $v_e < \infty$ each guard $g_i \in S_g^v$ can be activated/deactivated depending on $v_e$.
	
	In \cite{O'Rourke:1987}, it is shown that for any triangulation graph $G$ of a simple polygon $P$ with $n \geq 10$ edges it is always possible to find a diagonal $d \in E(G)$ that partitions $P$ into $P_1$ and $P_2$ such that $P_1$ is the triangulation graph of a pentagon, hexagon, heptagon, octagon or nonagon, which we call \textit{basic polygons}. In \cite{O'Rourke:1987}, it is shown that any triangulation graph of a pentagon, hexagon or heptagon has one dominating diagonal. It is also shown that any triangulation graph of a nonagon has two dominating diagonals, and in general every triangulation graph of a polygon with $n \geq 5$ vertices has a set of dominating diagonals with size of at most $\lfloor n/4 \rfloor$.
	
	\noindent
	
	\begin{remark}
		\bf{Lemma \ref{lemma:5} and Theorem \ref{theorem:3} prove the existence of at most $\lfloor\frac{n}{3}\rfloor$ vertices which cover the polygon, and contain a subset of at most $\lfloor\frac{n}{4}\rfloor$ vertices which are end points of diagonal guards that provide mobile coverage. In contrast, Theorem $1$ in \cite{O'Rourke:1983} proves the existence of dominating diagonals in a polygon. Although, proofs of Lemma \ref{lemma:5} and Theorem \ref{theorem:3} are inductive, the intermediate steps are different from the proof of Theorem $1$ in \cite{O'Rourke:1983}.}
	\end{remark}

	\begin{lemma}
		\label{lemma:5}
		For any given triangulation graph of a basic polygon, there exists a set of at most $\lfloor\frac{n}{3}\rfloor$ vertices which dominate the triangles of the triangulation, and contain a subset of at most $\lfloor\frac{n}{4}\rfloor$ vertices which are endpoints of diagonals which dominate the triangles of the triangulation. (There exists a set of at most $\lfloor\frac{n}{3}\rfloor$ candidate vertices).
	\end{lemma}
	\begin{proof} 
		\begin{enumerate}
			\item For a pentagon, the proof of the lemma follows from the proof of \cite{O'Rourke:1987} (Lemma $3.3$,page $85$) which states that there is always a vertex that can dominate any triangulation graph of a pentagon, so it can always be dominated by one diagonal guard with such a vertex as an endpoint of its diagonal while the opposite endpoint can be any other vertex of the triangulation. Hence, $|S_c|=|S_h|=1 \leq \lfloor 5/3 \rfloor$.
			
			\item Irrespective of the triangulation of the hexagon, it can always be decomposed into a pentagon $P_1$ and a triangle $T_1$ such that $T_1$ can be dominated by a vertex $v_1$ which is an arbitrary vertex of $T_1$. Therefore, a candidate vertex $v_0$ exists from which $P_1$ can be dominated according to the case of a pentagon. We can always select a diagonal with an endpoint at $v_0$ and the other at one vertex of $T_1$. Thus, $|S_c| \leq 2 = \lfloor 6/3 \rfloor$ and $|S_h| = 1=\lfloor 6/4 \rfloor$.
			
			\item Regardless of its triangulation, a heptagon can always be decomposed into a hexagon $H_1$ and a triangle $T_1$. Based on the case of a hexagon, there are two candidate vertices $v_0$ and $v_1$ from which $H_1$ can be dominated. Since $v_1$ can be an arbitrary vertex, a vertex of $T_1$ can always be chosen to cover it. $v_1$ is a candidate vertex on an arbitrary edge. The diagonal is selected such that one of its endpoints is $v_0$, and its other endpoint is a vertex of $T_1$. Thus, $|S_c| \leq 2 = \lfloor 7/3 \rfloor$ and $|S_h| = 1=\lfloor 7/4 \rfloor$.
			
			\item Regardless of its triangulation, an octagon can always be decomposed into a heptagon $H_2$ and a triangle $T_1$. Based on the case of a heptagon, there are two candidate vertices $v_0$ and $v_1$ from which $H_2$ can be dominated. Since $v_1$ can be a vertex of an arbitrary edge of $H_2$, one of the endpoints of the edge shared by $H_2$ and $T_1$ can always be chosen as a candidate vertex. One diagonal is selected such that one of its endpoints is $v_0$ and it dominates $H_2$. If it cannot dominate $T_1$, any diagonal with $v_1$ as an endpoint can be chosen to dominate $T_1$. Thus, $|S_c| \leq 2 = \lfloor 8/3 \rfloor$ and $|S_h| \leq 2 =\lfloor 8/4 \rfloor$.
			
			\item Regardless of its triangulation, a nonagon can always be decomposed into the triangulation graph of an octagon $O_1$ and a triangle $T_1$. From the case of an octagon, there are two candidate vertices $v_0$ and $v_1$ that dominate $O_1$. Hence, an additional guard $v_2$ is required to cover $T_1$. Since $v_1$ can be an arbitrary vertex of $T_1$ which in turn can be arbitrarily chosen to include any vertex of the boundary of the nonagon, it follows that $v_1$ can be arbitrarily placed, and the remaining triangulation subgraph of the nonagon has two candidate vertices. A diagonal incident\footnote{A diagonal is said to be \textit{incident} to a vertex if the vertex is an endpoint of the diagonal.} to $v_2$ can dominate at least two triangles of $G$. The remaining subgraph of $G$ is the triangulation graph of a heptagon with two candidate vertices ($v_0$ and $v_1$), and it can be dominated by a single diagonal guard with an endpoint incident to $v_0$ or $v_1$ (the case of a heptagon). Thus, $|S_c| \leq 3 = \lfloor 9/3 \rfloor$ and $|S_h| \leq 2=\lfloor 9/4 \rfloor$.
		\end{enumerate}
	\end{proof}

	\begin{theorem}
		\label{theorem:3}
		Every triangulation graph $G$ of a polygon with $n \geq 5$ vertices can be dominated by at most $\lfloor n/3 \rfloor$ candidate vertices. 
	\end{theorem}
	\begin{proof}
		Lemma \ref{lemma:5} proves the statement of the theorem for $5 \leq n \leq 9$. Therefore, we will prove the theorem for the case when $n \geq 10$ using induction. Let us assume that the theorem holds for any $n' < n$ and $n\geq10$. In \cite{O'Rourke:1987}, it is shown that a diagonal $d$ that partitions $G$ into two graphs $G_1$ and $G_2$, where $G_1$ contains $k$ boundary edges with $4 \leq k \leq 8$, always exist. We consider each value of $4 < k\leq 8$:

		\begin{enumerate}
			\item $k=5$: $G_1$ has $k+1=6$ boundary edges including $d$. By lemma \ref{lemma:5}, $G_1$ has two candidate vertices $v_0$ and $v_1$, with $v_1$ located arbitrarily. There is a dominating diagonal with $v_0$ or $v_1$ as an endpoint. $G_2$ has $n-k+1=n-4$ boundary edges including $d$, and since $v_1$ is arbitrary, it can be an endpoint of the edge between $G_1$ and $G_2$. Hence, we can always select $v_1$ such that at least two triangles $T_a$ and $T_b$ of $G_2$ are incident to $v_1$. Let $S_v(T_a \cup T_b)$ and $S_e(T_a \cup T_b)$ be the sets of vertices and edges of the union of $T_a$ and $T_b$ respectively. A graph $G_0$ can be constructed such that $V(G_0)=V(G_1) \cup S_v(T_a \cup T_b)$ and $E(G_0)= E(G_1) \cup S_e(T_a \cup T_b)$ (we say that $T_a$ and $T_b$ were adjoined to $G_1$), $G_2$ has now $n-6$ vertices which implies that it has at most $\lfloor n/3 \rfloor - 2$ candidate vertices. Thus, $G_0$ together with $G_2$ has at most $\lfloor n/3 \rfloor$ candidate vertices.

			\item$k=6$: This is similar to the previous case, but $G_1$ has $k+1=7$ boundary edges and $v_1$ is the endpoint of an arbitrary edge. $G_2$ has $n-k+1=n-5$ boundary edges including $d$. Since $v_1$ is the endpoint of an arbitrary edge, it can belong to the edge between $G_1$ and $G_2$. Hence, $v_1$ is selected such that it dominates at least one triangle of $G_2$. $G_0$ is obtained by adjoining this triangle to $G_1$, so $G_2$ has $n-6$ vertices. Therefore, it has at most $\lfloor n/3 \rfloor - 2$ candidate vertices. The result follows.
			
			\vspace{0.015in}
			\item $k=7$: $G_1$ has $k+1=8$ boundary edges. $G_2$ has $n-7+1=n-6$ boundary edges. In \cite{O'Rourke:1983} (Theorem $1$,page $280$) there are two cases: (i) the triangulation graph of the octagon has a single dominating diagonal. Therefore, $v_0$ and $v_1$ can be the endpoints of the diagonal. $G_2$ has $n-6$ boundary edges. Hence, it has $\lfloor n/3 -2 \rfloor$ candidate vertices, and $G$ has at most $\lfloor n/3\rfloor$ candidate vertices. (ii) A graph $G_3$ is obtained by adjoining three triangles of the triangulation graph of the octagon to $G_2$. The remaining subgraph corresponds to a pentagon. Since $G_3$ has $n-4+1=n-3$ vertices, it has $\lfloor n/3 \rfloor -1$ candidate vertices. Therefore, $G$ has at most $\lfloor n/3 \rfloor$ candidate vertices since the triangulation of the remaining pentagon has only one candidate vertex \ref{lemma:5}.
			
			\item $k=8$: $G_1$ has $k+1=9$ boundary edges, and by Lemma \ref{lemma:5}, it has three candidate vertices $v_0$, $v_1$ and $v_2$ that dominate $G$ such that $v_2$ is an arbitrary endpoint of the edge between $G_1$ and $G_2$, and it is an endpoint of one of the two dominating diagonals. $v_2$ can be chosen to cover at least two triangles of $G_2$. $G_0$ is obtained by adjoining those two triangles to $G_1$. Consequently, $G_2$ has $n-9$ boundary edges. Hence, $G_2$ has $\lfloor n/3 \rfloor -3$ candidate vertices, it follows that $G$ has at most $ \lfloor n/3 \rfloor$ candidate vertices.
		\end{enumerate}
		
		It is guaranteed that for any simple polygon with $n \geq 10$ vertices, there is always a diagonal $d$ that partitions the triangulation graph of the polygon such that the partition corresponds to one of the aforementioned cases. Since for each one of such cases it was shown that there exists a set of at most $\lfloor n/3 \rfloor$ candidate vertices, and by Lemma \ref{lemma:5} we know that the statement is always true for $4<n<10$, it follows that any triangulation graph of a simple polygon with $n>4$ vertices can always be dominated by at most $\lfloor n/3 \rfloor$ candidate vertices.
	\end{proof}
	
	Theorem \ref{theorem:3} proves the existence of a set of at most $\lfloor n/3 \rfloor$ candidate vertices that dominates $G$. Based on the different cases of the theorem and the existence of the diagonals that partition $G$ into subgraphs corresponding to basic polygons Algorithm \ref{alg:guards} is a strategy that we propose to determine the set of candidate vertices $S_c$ and its corresponding subset of dominating diagonals $S_h$. The objective of Algorithm \ref{alg:guards} is to iteratively partition $G$ into triangulation subgraphs of basic polygons, and once that $G$ is partitioned Algorithm \ref{alg:guards} exhaustively finds the sets of candidate vertices and dominating diagonals for each subgraph of the partition. At each iteration, it searches for a diagonal $d$ that separates a triangulation subgraph of a basic polygon (denoted by $G_p$) such that there is no other diagonal in $E(G_p)$ that can separate a subgraph of a smaller basic polygon. Also, for each subgraph $G_p$ identified, a set of candidate vertices in $V(G_p)$ is found such that they can cover all the triangles in $T(G_p)$ that are not already covered by other candidate vertices. It also finds a set of diagonals that can cover the triangles in $T(G_p)$ that are not already covered by other diagonals (and each diagonal must have a different candidate vertex as an endpoint). The process is repeated until the remaining non-partitioned subgraph has $9$ vertices or less. The first \textbf{while} cycle (Line $5$) finds the diagonal $d$ that partition $G$ into triangulation subgraphs that correspond to basic polygons. This can be completed in $O(n)$ time by traversing the dual graph $G_D$. Recall that $|V(D)|=n-2$ and $D$ is a tree. In Figure \ref{Fig:examp_dep} (a) a simple polygon is shown, it has $n=18$ vertices and $16$ triangles, the green segments represent the diagonals $d_1$, $d_2$ and $d_3$ found by Algorithm \ref{alg:guards}. $d_1$ separates the triangulation graph of the hexagon that contains triangles $T_5$, $T_6$, $T_7$ and $T_8$. $d_2$ (along with $d_1$) separates an hexagon containing $T_1$, $T_2$, $T_3$ and $T_4$. Finally, $d_3$ is found, which separates a heptagon containing triangles $T_9$, $T_{10}$, $T_{11}$, $T_{12}$ and $T_{13}$. The remaining subgraph corresponds to a pentagon. Figure \ref{Fig:examp_dep} (b) shows the corresponding dual graph, the blue vertices correspond to the triangles of the hexagon separated by $d_1$, the orange vertices to the triangles separated by $d_2$, and the green vertices to the triangles separated by $d_3$.
	
	\begin{figure}[thpb]
		\begin{center}
	\includegraphics[width=1\linewidth,height=0.51\linewidth]{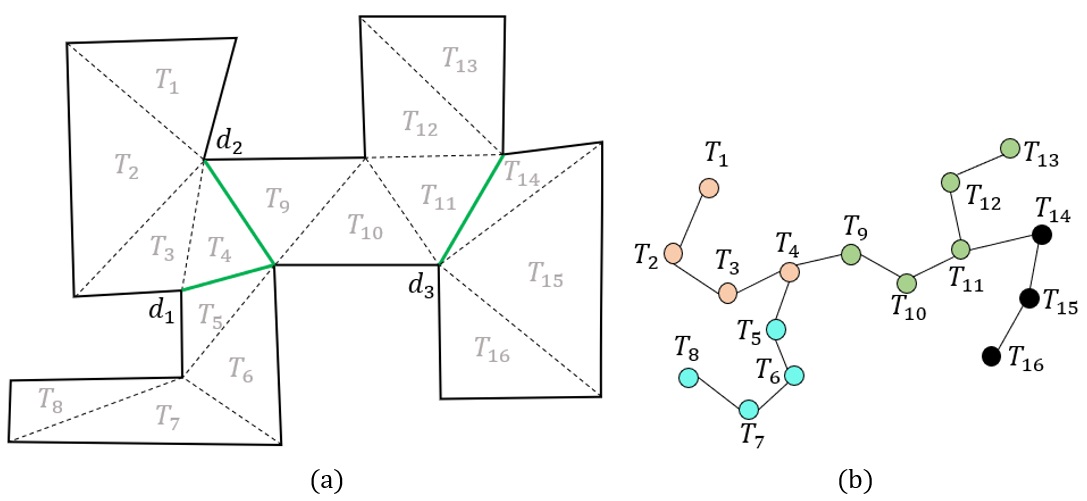}
		\end{center}
		\caption{(a) Triangulation graph of a polygon partitioned into triangulations of basic polygons. (b) Corresponding dual graph.}
		\vspace{-0.1in}
		\label{Fig:examp_dep}
	\end{figure}

	After the first \textbf{while} loop, a graph $G_{pol}$ is created (line $10$) such that each vertex $v \in V(G_{pol})$ corresponds to the triangulation graph of each basic polygon $G_1$ that was found in line $6$, and there is an edge $e \in E(G_{pol})$ only between vertices that correspond to basic polygons that have a common edge in $G$. For this step, $D$ along with $S_D$ can be used to identify in linear time the triangles in $T(G)$ that correspond to the basic polygons of the partition, and consequently their subgraphs $G_1$. This is done in $O(n)$ time. In Figure \ref{Fig:examp_dep2} (b) the graph $G_{pol}$ that corresponds to example of Figure \ref{Fig:examp_dep} is shown, the blue and orange vertices correspond to the hexagons found in the partition, the green vertex corresponds to the heptagon and the black vertex corresponds to the remaining pentagon. The second \textbf{while} loop identifies the minimum set of vertices that can cover the triangles in $T(G_i)$ that are not already covered by other candidate vertices (line $14$). Finding such a subgraph $G_i$ takes $O(n)$ time. It also finds the corresponding minimum set of diagonals that can dominate the triangles in $T(G_i)$ that are not already covered by other diagonals such that each one of them has a different candidate vertex as an endpoint (line $15$). These steps can be completed exhaustively in constant time since the triangulation graph $G_i$ always corresponds to a basic polygon. Since $|V(G_{pol})|< n$, the second \textbf{while} loop clearly takes $O(n)$ time. It follows that the time complexity of the algorithm is $O(n^2)$. The procedure to mark the visited vertices in $V(G_{pol})$ ensures that at any time a triangulation subgraph is selected by the algorithm to find its corresponding candidate vertices and dominating diagonals, there is at most one neighboring vertex in $V(G_{pol})$ that corresponds to a triangulation subgraph such that not all of its dominating diagonals and candidate vertices have already been found. The second \textbf{while} loop obtains a set of at most $\lfloor n/4 \rfloor$ diagonal guards using the basic polygons represented in $G_{pol}$. Once that $S_h$ and $S_c$ are found, diagonal guards are deployed in the diagonals of $S_h$ while vertex guards are deployed in the vertices $S_c$ that are not associated to diagonals in $S_h$.	
	
	\begin{figure}[thpb]
		\begin{center}
			\includegraphics[width=0.9\linewidth,height=0.51\linewidth]{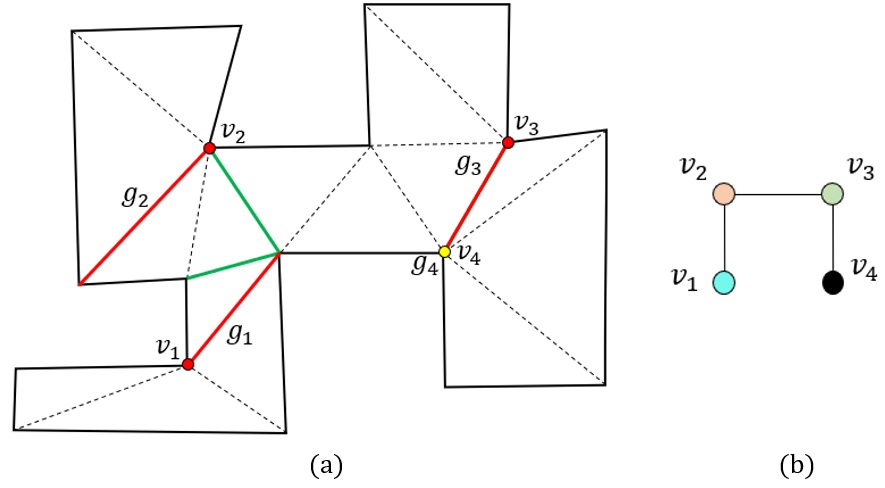}
		\end{center}
		\caption{(a) Candidate vertices and dominating diagonals. (b) Corresponding graph $G_{pol}$.}
		\vspace{-0.1in}
		\label{Fig:examp_dep2}
	\end{figure}
	
	In Figure \ref{Fig:examp_dep2} (a) the red diagonals represent the dominating diagonals, the red shaded circles correspond to the candidate vertices assigned to the dominating diagonals and the yellow shaded circle corresponds to the remaining candidate vertex which is not associated to any diagonal. Starting with the hexagon defined by triangles $T_8$, $T_7$, $T_6$ and $T_5$ there is a single vertex from which all the triangles of the hexagon can be covered ($v_1$). Any diagonal with $v_1$ as an endpoint dominates the hexagon. However, the diagonal chosen is the one that dominates more triangles (in the example it is the diagonal that dominates $T_4$ $T_9$ and $T_{10}$). Next, the algorithm proceeds with the second hexagon, where $v_2$ is found and since any diagonal of the hexagon with $v_2$ as an endpoint dominates the hexagon, the diagonal labeled as $g_2$ was arbitrarily selected. The algorithm proceeds with the triangulation graph of the heptagon where two candidate vertices $v_3$ and $v_4$ are found, and also one diagonal incident to both which is able to dominate the triangles that are not already dominated. Hence, $S_c=\{ v_1, v_2, v_3, v_4 \}$ and $S_h=\{ h_1, h_2, h_3 \}$. Clearly,  $|S_c| < \lfloor 18/3 \rfloor$ and $|S_h| < \lfloor 18/4 \rfloor$. Finally, three diagonal guards $g_1$, $g_2$ and $g_3$ are deployed along $h_1$, $h_2$ and $h_3$ respectively, and a vertex guard $g_4$ is located at $v_4$.


	\begin{algorithm}
		\caption{Guard Deployment}
		\begin{algorithmic}[1]
			\State\textbf{Input}: $G$.
			\State\textbf{Output}: $S_h$ and $S_c$
			\State $S_D \leftarrow \emptyset$ is the set of diagonals $d$
			\State $G' \leftarrow G$
			\While{$G'$ has $n \geq 10$ vertices}
			\State find $d$ that separates the triangulation of a minimal basic polygon $G_1$ from $G'$
			\State $G'$ becomes the subgraph obtained by removing $G_1$ excepting the vertices of $d$
			\State add $d$ to $S_D$
			\EndWhile
			\State create $G_{pol}$ from $G$ using the diagonals in $S_D$
			\While{there is an unmarked vertex in $G_{pol}$}
			\State $v_i \leftarrow$ unmarked vertex in $G_{pol}$ with at most one unmarked neighbor
			\State $G_i \leftarrow$ subgraph of $G$ that corresponds to $v_i$
			\State $s_c \leftarrow$ appropriate candidate vertices of $G_i$
			\State $s_h \leftarrow$ appropriate diagonals of $G_i$ 
			\State add vertices of $s_c$ to $S_c$
			\State add diagonals of $s_h$ to $S_h$
			\State mark $v_i$
			\EndWhile
		\end{algorithmic}
		\label{alg:guards}
	\end{algorithm}
	
	\section{Critical Regions and Motion Strategy}
	\label{sec:crit}

	In this section we present a motion strategy for the mobile guards for persistent tracking. We also classify the triangles of the triangulation of $G$ depending on the diagonals of the guards incident to them. A classification of guards is also presented based on the types of triangles that are incident to the endpoints of their diagonals, and a condition that guarantees tracking when the intruder is inside a given triangle is also presented.
	
	Let $T_j(i) \subset T(G)$ be the set of triangles incident to endpoint $v_j(i)$. We say that the set of triangles $T(i) = T_1(i) \cup T_2(i)$ is ``incident'' to $h_i$. Additionally, we define $d_M^i=r l_i$. It is the ``Maximum'' distance that the intruder can travel while $g_i$ moves across $h_i$. The triangles of $T(G)$ can be classified in the following categories based on the relative location of the diagonals in $S_h$ (refer to Figure \ref{fig:classtriang}):
	
	\begin{enumerate}
		\item A triangle is called \textit{safe} if it can be covered at any time the intruder is inside it. They are illustrated as light blue shaded triangles in Figure \ref{fig:classtriang}. Trivially, a triangle such that one or more of its edges is a diagonal in $S_h$, or a triangle such that there is a static guard located at any of its vertices is a safe triangle. In Figure \ref{fig:classtriang}, $g_3$ is a static guard located at the vertex represented as a red circle and all the triangles that have such a vertex in common are safe triangles. 
		\item A triangle is called \textit{unsafe} if it is not a safe triangle and there is only one guard that can cover it from one of its endpoints. They are illustrated as orange shaded triangles in Figure \ref{fig:classtriang}.
		\item A triangle is called \textit{regular} if it is neither safe nor unsafe. They are illustrated as unshaded triangles in Figure \ref{fig:classtriang}. Notice that for each regular triangle there are at least two dominating diagonals (red segments) incident to its vertices.
		\item Safe Zone: A \textit{safe zone} of a guard $g_i$, denoted by $A(i)$, is the set of triangles adjacent which have $h_i$ as an edge. Notice that all the triangles in a safe zone are safe triangles, see Figure \ref{fig:classtriang}.
		\item Augmented Safe Zone: An augmented safe zone $B(i)$ is the largest set of adjacent safe \footnote{Two triangles are adjacent if they share an edge.} in $T(i)$ such that $A(i) \subseteq B(i)$
		\item Unsafe Zone: An \textit{unsafe zone} of an endpoint $v_j(i)$ of $h_i$, with $j \in \{1,2\}$,is the set of unsafe triangles incident to $v_j(i)$ and it is denoted by $U_j(i)$. Also let $U(i)=U_1(i) \cup U_2(i)$ be the unsafe zone of $g_i$.
	\end{enumerate}
	
	\begin{figure} 
		\begin{center}
			\includegraphics[width=0.95\linewidth,height=0.54\linewidth]{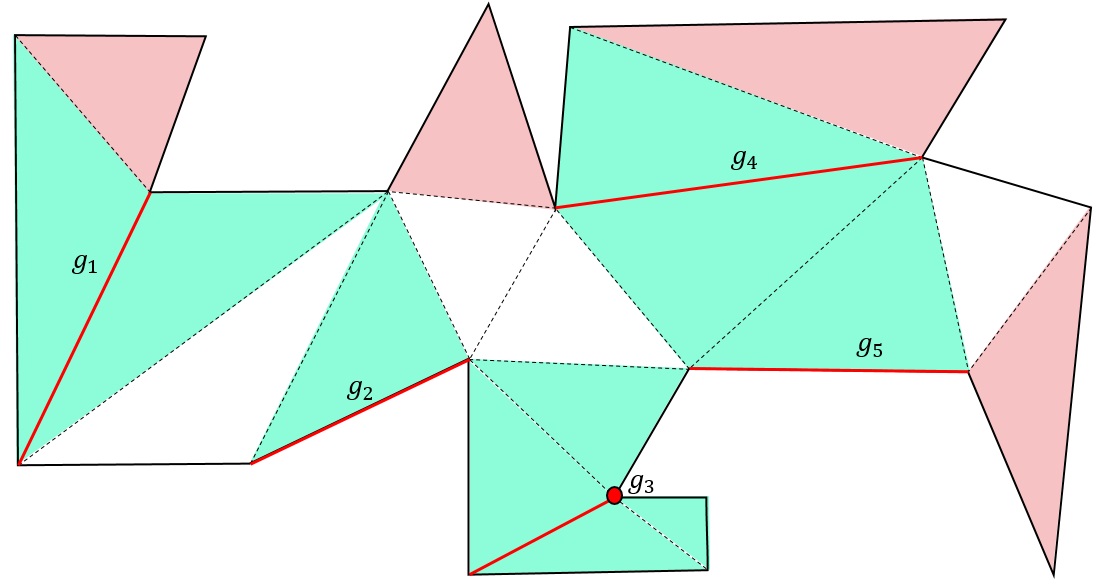}
		\end{center}
		\caption{Classification of triangles.}
		\label{fig:classtriang}
	\end{figure}
	
	\begin{enumerate}
		\item Type $0$ Guards: A guard $g_i \in S_g$ is of type $0$ if there there is an endpoint of $h_i$ such that all the triangles incident to such endpoint are safe triangles, which means that each one of them (excepting the triangles in $A(i)$) is covered by a guard $g_{k} \in S_g \backslash \{g_i\}$ when $p_I \in T$. Hence, a Type $0$ guard does not require to be located at one endpoint of its diagonal, so it can just stay motionless at the opposite endpoint, labeled as $v_1(i)$. An example of a Type $0$ guard is $g_3$ in Figure \ref{fig:classguard}.
		
		\item Type $1$ Guards: A guard $g_i$ (not of Type $0$) is of Type $1$ if it has an unsafe zone $U_k(i)$, and there are no regular triangles in $T_k(i)$ adjacent to $B(i)$. In such a case, we label $k=1$. For Type $1$ guards, we define an internal critical curve, denoted by $S_{int}^1(i)$, as the boundary between $U_1(i)$ and $B(i)$. $S_{int}^1(i)$ partitions $P$ into two regions. We define a region $R= \bigcup_{T_j \in U_1(i)} T_j$, and let $P_R$ denote the partition containing $R$. We define an external critical curve, $S_{ext}^1(i)$, as the curve inside $P \backslash P_R$ which is at a distance of $d_M^i$ from $S_{int}^1(i)$. The critical region $C_1(i)$ is defined as the region inside $P$ ``enclosed'' by $S_{int}^{1}(i)$ and $S_{ext}^{1}(i)$. Depending on the location of $I$ inside $C_1(i)$, $g_i$ is located at a different point at the interior of $h_i$, see Equation (\ref{eq:1}). Given the definition of $S_{ext}^{1}(i)$ and $d_M^i$, the critical curves can be used to trigger the motion of $g_i$ such that $g_i$ starts moving from one endpoint of $h_i$ to the other when $I$ enters $C_1(i)$, and it is ensured that $g_i$ can reach the opposite endpoint of $h_i$ at the same time that $I$ reaches the other critical curve bounding $C_1(i)$. A Type $1$ guard $g_1$ is illustrated in Figure \ref{fig:classguard} along with its critical curves $S_{int}^{1}(1)$ and $S_{ext}^{1}(1)$ shown as blue curves.
		
		\item Type $2$ Guards: 	A guard $g_i \in S_g$ (not of Type $0$ or Type $1$) is of Type $2$ if all the neighboring guards\footnote{We say that a guard $g_k \in S_g\backslash \{g_i\}$ is a \textit{neighbor} of $g_i$ if $T(i) \cap T(k) \neq \emptyset$.} that can cover the regular triangles incident to one endpoint of $h_i$ have their critical regions already defined. Such an endpoint is then labeled as $v_1(i)$. In Figure \ref{fig:classguard}, $g_1$ has its critical region defined since it is a Type $1$ guard. Moreover, since $g_1$ is the only neighbor of $g_2$ that can cover a regular triangle incident to one endpoint of the diagonal of $g_2$, then $g_2$ meets the definition of a Type $2$ guard. Define $R_1(i) \subset T_1(i)$ as the set of regular triangles incident to $v_1(i)$, and define $N_1(i)$ as the set of all neighboring guards that cover triangles in $R_1(i)$. To generate the critical region $C_1(i)$ of a type $2$ guard $g_i$, each $T_j \in R_1(i)$ is considered. We define $S_{T_j} \subset N_1(i)$ as the set of guards that can cover $T_j$. Let $B_j = T_j\cap (\bigcap_{g_l \in S_{T_j}}{C_1(l)})$. If $B_j \neq \emptyset$ and if $p_I \in B_j$, then $g_i$ is the only guard that can cover $B_j$, since, by the definition of a critical region, there is no guard in $S_{T_j}$ that is covering $T_j$. Let $S_B$ be the region obtained from the union of all regions $B_j$ and all the unsafe triangles incident to $v_1(i)$. We define $S_{int}^{1}(i)$ as the boundary of $S_B$. $S_{ext}^{1}(i)$ and $C_1(i)$ are defined in the same manner that they are defined for Type $1$ guards. $S_{ext}^{1}(i)$ is the curve inside the region $P \backslash S_B$ which is at a distance of $d_M^i$ from $S_{int}^{1}(i)$, and $C_1(i)$ is then the region inside $P$ ``enclosed'' by $S_{int}^{1}(i)$ and $S_{ext}^{1}(i)$. In Figure \ref{fig:classguard}, $S_{int}^{1}(2)$ is the closed black curve inside the regular triangle shared by $g_1$ and $g_2$, and $S_{ext}^{1}(2)$ is shown as a black curve inside $P$ that maintains a constant distance from  $S_{int}^{1}(2)$.
		
		\item Type $3$ Guards: A guard $g_i \in S_g$ is a type $3$ guard if it is not of Type $0$,$1$ or $2$. A distinctive characteristic of this type is that there is at least one regular triangle incident to each endpoint of $h_i$ such that it is adjacent to $B(i)$, and there is at least a pair of neighbors without their critical regions defined. Hence, we cannot determine all the regions $R$ that would define the internal critical curve of a Type $2$ guard. The presence of Type $3$ guards is associated with the presence of a cyclic arrangement of Type $3$ guards and regular triangles as shown in Figure \ref{fig:path} (a). Clearly, there is an ambiguity in allocation for those regular triangles since the region that must be assigned to $g_i$ cannot be found as in the cases of Type $1$ and Type $2$ guards. Thus, for Type $3$ guards, we proceed by transforming all the regular triangles incident to one of the endpoints of its diagonal into unsafe triangles, the selected endpoint is then labeled as $v_1(i)$. This turns $g_i$ into a Type $1$ guard, so a critical region for it can be obtained. Consequently, $g_i$ has the task of covering all the triangles in $T_1(i)$. Since all triangles in $R_1(i)$ are unsafe triangles that must be covered by $g_i$, they are considered as ``safe triangles'' for the neighboring guards. The change of the types of triangles, changes the types of the neighboring guards(Type $3$ guards may transform into any of the other types, Type $2$ guards may become Type $0$ or Type $1$ guards, and Type $1$ guards may become Type $0$ guards). In Figure \ref{fig:classguard}, two Type $3$ guards $g_4$ and $g_5$ are illustrated. Triangles $T_3$ and $T_4$ are regular triangles shared by both type $3$ guards, so it is not possible to obtain the internal critical curve of $g_4$ and $g_5$. Hence, $T_4$ is arbitrarily redefined as an unsafe triangle for $g_5$ (and consequently as a safe triangle for $g_4$), that means that it is assigned to $g_5$. As a consequence, $g_5$ meets the definition of a Type $1$ guard, and since $T_4$ is considered a safe triangle by $g_4$, then it also meets the definition of a Type $1$ guard. The resulting critical curves are also illustrated.	
	\end{enumerate}
	
	\begin{figure} 
		\begin{center}
			\includegraphics[width=0.95\linewidth,height=0.54\linewidth]{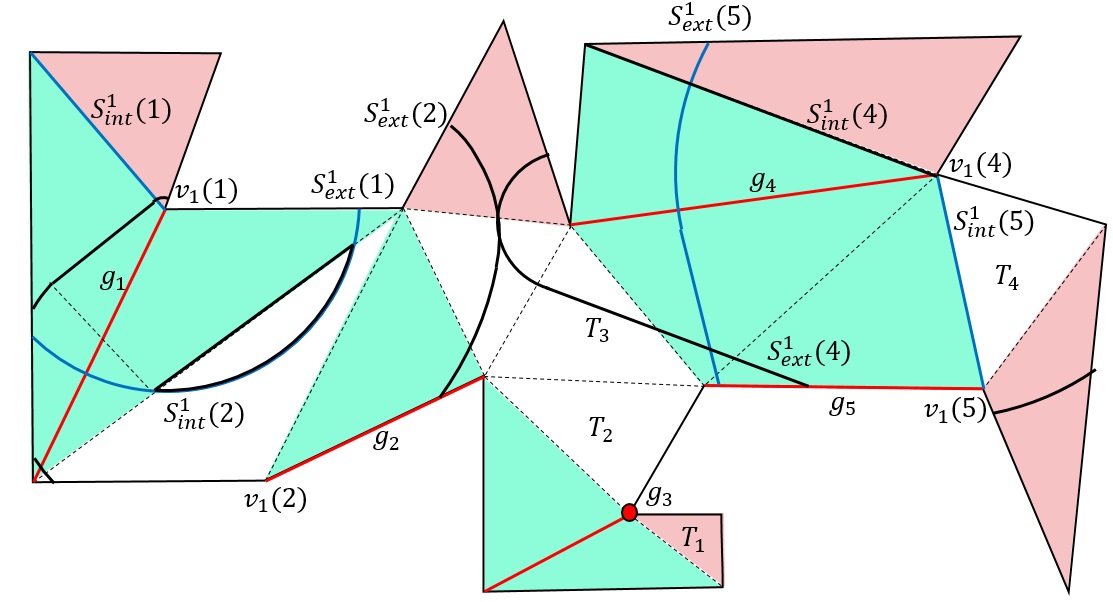}
		\end{center}
		\caption{Example of all the types of guards in the same environment: $g_1$ is a Type $1$ guard, $g_2$ is a Type $2$ guard, $g_3$ is a Type $0$ guard, and $g_4$ and $g_5$ are Type $5$ guards.}
		\label{fig:classguard}
	\end{figure}
	
	\subsection{Reactive Strategy}
	
	Let $p_I \in C_1(i)$ and let $d_{min}^1(i)$ be the minimum distance between $p_I$ and $S_{int}^{1}(i)$. The following equation maps the location of the intruder to obtain the location of a guard, $p_{g_i}$, along its diagonal:	
	\begin{equation}
	\label{eq:1}
	p_{g_i}= p_{v_1(i)}+ \frac{d_{min}^1(i)}{d_M^i}(p_{v_2(i)}-p_{v_1(i)}),
	\end{equation}
	where $p_{v_1(i)}$ and $p_{v_2(i)}$ are the coordinates of $v_1(i)$ and $v_2(i)$ respectively. If $p_I \notin C_1(i)$, $g_i$ remains static at $v_1(i)$ or $v_2(i)$ depending on the location of $I$.
	
	We present an important result that determines the condition under which non-safe (regular or unsafe) triangles can be turned into safe triangles.
	
	\begin{lemma} 
		\label{lemma:10}
		A non-safe triangle $T$ is covered if and only if $\bigcap_{g_i \in S_g(T)} C_1(i) \cap T = \emptyset$.
	\end{lemma}
	\begin{proof} $\Rightarrow$ If $T$ is covered, a guard is always present at its vertex in case $I$ lies inside it. If $I$ lies inside the critical region associated with a guard, the guard cannot lie on the vertex of $T$. Therefore, the critical regions associated with the guards cannot have a common intersection inside $T$ since it is covered. $\Leftarrow$ If $\bigcap_{g_i \in S_g(T)} C_1(i) \cap T = \emptyset$, every point inside $T$ lies outside the critical region of at least one $g_i \in S_g(T)$. Therefore, $T$ is covered. 
	\end{proof}

	\section{Activation of Additional Guards}
	\label{sec:activ}
	\vspace{-0.05in}
	In this section, we present a procedure to activate or deactivate the guards to ensure tracking depending on the instantaneous value of $r$. Once $S_g$ and $S_c$ are defined and the critical regions of the guards in $S_g$ are obtained, we want to determine the number of activated guards based on $v_e$. First, we introduce some notation. Given $T \in T(G)$, let $V(T)$ and $E(T)$ be the sets of vertices and edges of $T$ respectively. $S_v^{ac},S_v^{in} \subset S_g^v$ denote the set of activated and deactivated vertex guards respectively. $S_v^{di} \subseteq S_c$ is the set of candidate vertices associated to the diagonals in $S_h$. Clearly, $|S_c|=|S_v^{di}|+|S_v^{ac}|+|S_v^{in}|$. Since $S_h$ dominates $G$, the polygon is covered by $S_g$ as $v_g \rightarrow \infty$. Therefore, the guards in $S_g$ are sufficient to track the intruder at $r=0$. However, some vertex guards might need to be activated as $r$ increases.
	
	\begin{figure} [thpb]
		\begin{center}
			\includegraphics[width=0.93\linewidth,height=0.82\linewidth]{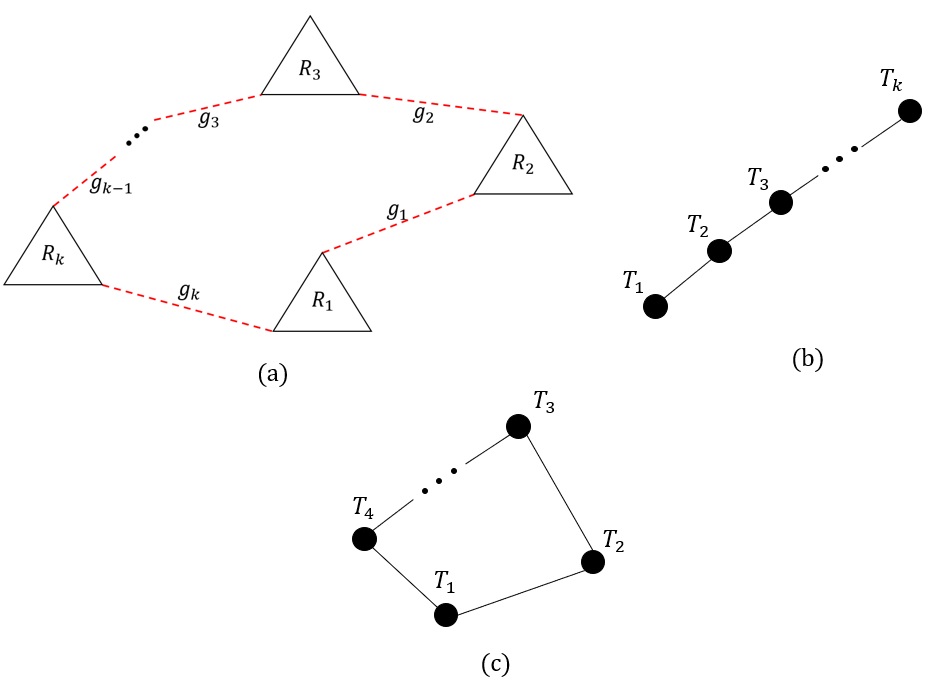}
		\end{center}
		\caption{(a) The presence of type $3$ guards means that a closed path between regular triangles and diagonals of guards can be found in $G$. (b) Graph that represents the recursive process of Algorithm \ref{alg:activation}. (c) The graph cannot have a closed path.}
		\label{fig:path}
	\end{figure}
	
	We define $S_T' \subset T(G)$ as the set of non-safe triangles of $G$. Algorithm \ref{alg:update} shows the pseudocode to activate guards based on the condition in Lemma \ref{lemma:10}. The input to this algorithm is the polygon, the position of the guards, and the instantaneous $r$. At every iteration, Algorithm \ref{alg:update} searches for a triangle that cannot be covered ({\bf Lines $11$ and $19$}) based on Lemma \ref{lemma:10}. If such a triangle exists, Algorithm \ref{alg:activation} searches for a guard in $S_v^{in}$ to activate by the following recursive process on a graph $G_{rec}$ (refer to Figure \ref{fig:path} (b)). The algorithm starts with a non-safe triangle $T_c$, if $S_v^{in} \cap V(T_c) =\emptyset$, then it selects a $g_i \in S_g$ followed by a non-safe triangle at the other endpoint of $h_i$. For each non-safe triangle visited by the algorithm we define a vertex a $v \in V(G_{rec})$ and for each $g_i$ used in the algorithm we define an edge $v_1v_2 \in E(G_{rec})$. Thus, a graph $G_{rec}$ is built. It is a path that ends at a vertex $v$ that corresponds to a triangle $T_c$ where $S_v^{in} \cap V(T_c) \neq \emptyset$. 
	
	\begin{lemma} 
		\label{lemma:11}
		Algorithm \ref{alg:activation} terminates in finite steps. 
	\end{lemma}
	\begin{proof}
		The recursive procedure chooses a diagonal guard at each step. If the algorithm does not terminate, it implies that either $|S_g| = \infty$ (which is not possible), or a cycle exists in $G_{rec}$ (refer to Figure \ref{fig:path} (c). The latter cannot occur since the existence of such a cycle implies the existence of Type $3$ guards (refer to Figure \ref{fig:path} (a)). Since all Type $3$ guards are converted to Type $1$ guards before the algorithm starts (discussed in Section \ref{sec:crit}), $G_{rec}$ cannot have a cycle.
	\end{proof}
	
	At the end of each execution of Algorithm \ref{alg:activation}, at least one non-safe triangle is converted into a safe triangle, and this in turn reduces the area of $\bigcap_{g_i \in S_g(T_c)} C_1(i)$. Moreover, since the number of non-safe triangles is finite, the number of executions of Algorithm \ref{alg:activation} is finite. After each iteration of Algorithm \ref{alg:update}, a non-safe triangle $T$ gets converted to a safe triangle, and there are no regions inside $T$ that need to be covered by any diagonal guard. This change ``relieves'' the guards $g_i \in S_g(T)$ of the responsibility of covering $T$. Consequently, the critical regions $C_1(i)$ decrease their size. Moreover, smaller critical regions $C_1(i)$ increase the chance of the condition of Lemma \ref{lemma:10} to be met. Finally, the following lemma shows that Algorithm \ref{alg:update} terminates.

	Algorithm 3, forces lemma 2 to be satisfied for some guards, tracking is guaranteed (maybe this needs to be proved if the paper is self-contained), lemma 4 proves that algorithm 2 calls algorithm 3 (forces lemma 2 to be satisfied) any time that it is required so tracking is ensured for all the guards, that implies correctness. 
	
	\begin{lemma}
		Algorithm \ref{alg:update}, guarantees that the triangle where $I$ is located is always covered by a guard regardless of $v_e$.
	\end{lemma}
	\begin{proof}
		The proof is trivial for safe triangles since by definition, they are always covered when there is an intruder inside them according to Lemma \ref{lemma:10}. For the case of non-safe triangles, Algorithm \ref{alg:update} activates vertex guards for the current value of $v_e$ until the condition of Lemma \ref{lemma:10} is true for all non-safe triangles. Since Algorithm \ref{alg:activation} terminates in finite steps, and since after each iteration of Algorithm \ref{alg:update}, a non-safe triangle $T$ gets converted to a safe triangle, Algorithm \ref{alg:update} terminates after a finite number of steps. Since the condition of Lemma \ref{lemma:10} is met by all the triangles after a finite number of iterations of Algorithm \ref{alg:update}, the result follows.
	\end{proof}
	
	\begin{algorithm}
		\caption{UpdateActiveGuards}
		\begin{algorithmic}[1]
			\State\textbf{Input}: $P$,$G$,$r$,$S_T'$,$S_g^{di}$,$S_v^{in}$,$S_g$,$S_v^{ac}$ 
			\State\textbf{Output}: Updated $S_v^{ac}$ and $S_v^{in}$
			\If{$r$ increased}
			\State for each $g_i \in S_g$ compute $C_1(i)$
			\State $S_T^1 \leftarrow \{T \in S_T':$ condition of Lemma \ref{lemma:10} fails  $\}$
			\While{$|S_T^1|>0$}
			\State $T_c \leftarrow $ arbitrary $T \in S_T^1$
			\State \textbf{ActivateGuard}($T_c$,$S_T'$,$S_v^{di}$,$S_v^{in}$,$S_g$)
			\State update classification of guards in $S_g$
			\State for each $g_i \in S_g$ compute $C_1(i)$
			\State $S_T^1 \leftarrow \{T \in S_T':$ condition of Lemma \ref{lemma:10} fails $\}$
			\EndWhile
			\ElsIf{$r$ decreased}
			\While {$S_T^1= \emptyset$ and $S_v^{ac} != \emptyset$}
			\State $g \leftarrow$ arbitrary $g_i \in S_v^{ac}$
			\State deactivate $g$
			\State update classification of guards in $S_g$
			\State for each $g_i \in S_g$ compute $C_1(i)$
			\State $S_T^1 \leftarrow \{T \in S_T':$ condition of Lemma \ref{lemma:10} fails $\}$
			\EndWhile
			\State activate $g$
			\EndIf
		\end{algorithmic}
		\label{alg:update}
	\end{algorithm}

	\section{Results}
	
	\label{sec:examp}
	
	In this section, examples that illustrate the relation between different values of $r$ found through simulation and the number of active guards are presented. Figure \ref{fig:guards1} (a) shows a simple polygon with $n=19$ vertices. The candidate vertices (from Algorithm \ref{alg:guards}) are marked in red. Diagonal drawn in red represent the diagonal guards. $|S_c|=5<\lfloor 19/3 \rfloor=6$ with $|S_g|=3$, $S_v^{ac} = \emptyset$ and $|S_v^{in}|=2$. $g_3$ needs to cover triangles $7$ and $9$ from each endpoint. When $r>0$, the critical region of $g_3$ does not meet the condition of Lemma \ref{lemma:10} in triangle $9$. As a result, Algorithm \ref{alg:update} activates the inactive guard in triangle $9$. Activated vertex guards are marked in yellow with the corresponding $r$ that triggers their activation. After activation, triangles $9,10,15$ and $16$ become safe. It follows that $g_3$ meets the definition of a Type $0$ guard covering triangles $5,6,7$ and $17$ (so all of them become safe triangles). For $0<r \leq 0.8$, $|S_g \cup S_v^{ac}|=4$. For $r>0.8$, the condition of Lemma \ref{lemma:10} is violated in triangle $8$ (the critical regions of $g_1$ and $g_2$ intersect in $8$). Hence, Algorithm \ref{alg:update} finds the inactive guard in triangle $8$ and activates it, thereby, covering triangles $13,14,15,16$ and $8$. Consequently, $g_1$ and $g_2$ become Type $0$ covering the remaining triangles. The red plot in Figure \ref{fig:guards1} (b) shows the variation in the number of active static guards as $r$ increases. For any $r>0.8$, the total number of guards does not change since the environment is completely covered at that point. Persistent area coverage is achieved for any value of $r$ greater than $0.8$ which implies that even if the intruder has infinite speed the set of $3$ diagonal guards and $2$ vertex guards suffices to keep track of $I$.

	\begin{algorithm}
		\caption{ActivateGuard}
		\begin{algorithmic}[1]
			\State\textbf{Input}: $T_c$,$S_T'$,$S_g^{di}$,$S_v^{in}$,$S_g$
			\State\textbf{Output}: Updated $S_T'$, $S_v^{in}$
			\If{$S_v^{in} \cap V(T_c) \neq \emptyset$}
			\State $v_c \leftarrow$ arbitrary $v \in S_g^{in} \cap V(T_c)$
			\State activate guard in $v_c$
			\State update classification of triangles
			\Else
			\State $v_c \leftarrow$ arbitrary $v \in S_v^{di} \cap V(T_c)$
			\State $g \leftarrow$ $g_i \in S_g$ such that $v_c=v_j(i)$
			\State $T_c \leftarrow$ arbitrary $T_c \in T_k(i) \cap S_T'$
			\State \textbf{ActivateGuard}($T_c$,$S_T'$,$S_v^{di}$,$S_v^{in}$,$S_g$)
			\EndIf
		\end{algorithmic}
		\label{alg:activation}
	\end{algorithm}
	
	\begin{figure}[thpb]
		\begin{center}
			\subfigure[]{\includegraphics[width=0.63\linewidth,height=0.64\linewidth]{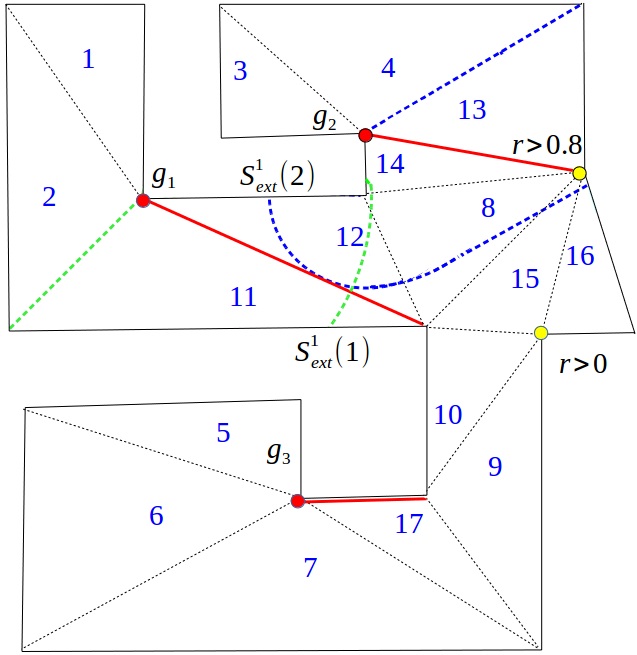}}
			\subfigure[]{\includegraphics[width=0.75\linewidth,height=0.59\linewidth]{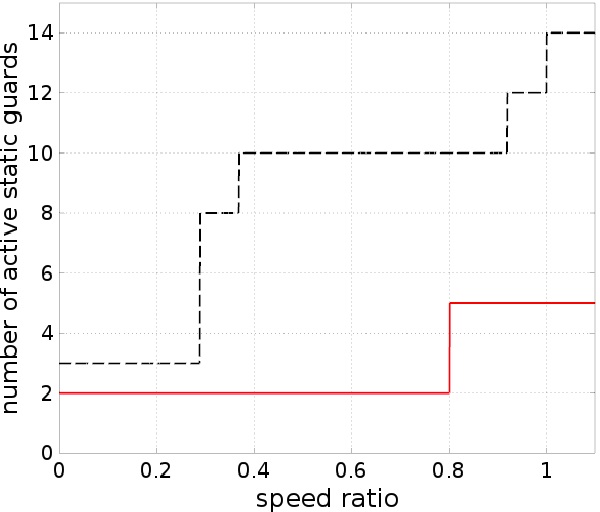}}
		\end{center}
		\caption{(a) Diagonal guards and static guards. Each static guard is labeled with the value of $r$ that triggers its activation. (b) The number of activated static guards increases as $r$ grows for examples $1$ (red solid line) and $2$ (dashed black line).}
		\label{fig:guards1}
	\end{figure}

	In Figure \ref{fig:exp2}, a simple polygon with $n=48$ vertices representing the environment is illustrated. $|S_c|=14<\lfloor 48/3 \rfloor=16$ with $|S_g|=9$, $S_v^{ac} = \emptyset$ and $|S_v^{in}|=5$. Since $g_5$ needs to cover triangles $39$ and $42$ from both endpoints, it meets the definition of a Type $0$ guard so it can cover triangles $39,42,40,41,44$. These triangles along with other triangles with a diagonal guard as an edge are safe triangles. $g_4$ needs to cover triangles $60$ and $76$ from each endpoint. When $r>0$, the critical region of $g_4$ is such that it causes the condition of Lemma \ref{lemma:10} to be violated in triangle $60$. Therefore, the inactive guard in triangle $36$ becomes active. After the activation, triangles $37$ and $36$ become safe. It follows that $g_4$ meets the definition of a Type $0$ guard covering triangles $31$ and $32$. 
	
	\begin{figure}[thpb]
		\begin{center} 
			\includegraphics[width=0.82\linewidth,height=0.65\linewidth]{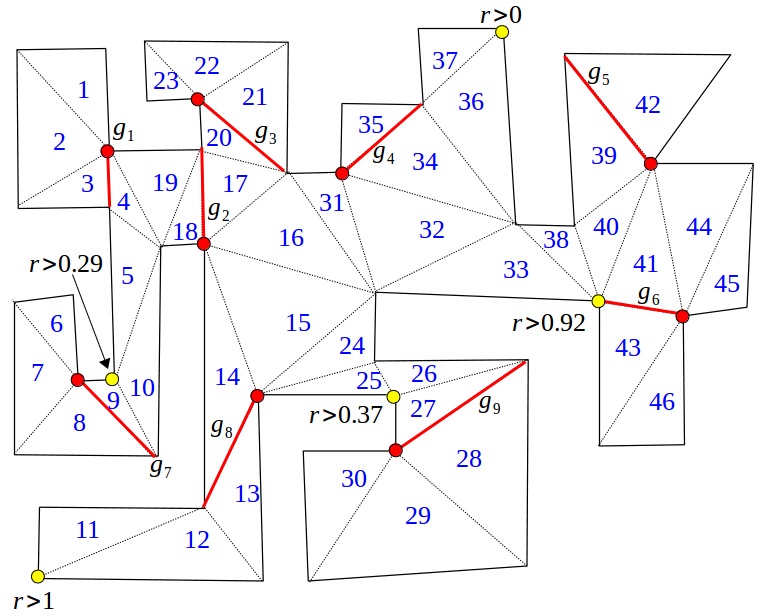}
		\end{center}
		\caption{Diagonal guards and static guards. Each static guard is labeled with the value of $r$ that triggers its activation.}
		\label{fig:exp2}
	\end{figure}
	
	For $0<r \leq 0.29$, $|S_g \cup S_v^{ac}|=10$. When $r>0.29$, the condition of Lemma \ref{lemma:10} is not met in triangle $10$. Hence, Algorithm \ref{alg:update} finds the inactive guard in triangle $10$ and activates it, so it covers triangles $5$, $10$ and $9$. Consequently, $g_7$ and $g_1$ become Type $0$. Therefore, $g_7$ covers triangles $6,7,8$ and $9$. $g_1$ covers triangles $1,2,3,4$ and $19$, which causes $g_2$ to also become Type $0$ covering triangles $14,15,16,17$ and $18$. For $0.29<r \leq 0.37$, $|S_g \cup S_v^{ac}|=11$. Nevertheless, when $r>0.37$, the condition of Lemma \ref{lemma:10} is not met in triangle $26$. Hence, Algorithm \ref{alg:update} activates the inactive guard located at a vertex of triangle $26$, so it covers triangles $25,26$ and $27$. $g_9$ becomes Type $0$ so it covers triangles $27,28,29$ and $30$. Following the same procedure, the inactive guard located at triangle $40$ is activated when $r>0.92$. When $r>1$, the condition of Lemma \ref{lemma:10} is violated in triangle $24$. Algorithm \ref{alg:update} does not find an inactive guard so it uses $g_8$ to find a triangle where there is an inactive guard that can be activated. The inactive guard in triangle $12$ is activated, and consequently, $g_8$ becomes a Type $0$ covering triangles $13,14,15,24$ and $25$. The dashed black plot in Figure \ref{fig:guards1} (b) shows the variation in the number of active static guards as $r$ increases. For any $r>1$, the total number of guards does not change since the environment is completely covered at that point. As it has been shown, the number of activated guards depends on $r$, so for low speed regimes of the intruder, not all guards need to be activated, in contrast with the case of only static guards where the total number of activated guards for both examples would be $5$ and $14$ respectively regardless of $r$.

	\section{Conclusions}
	\label{sec:conclusion}
	
	In this paper, we explored a problem in which a team of static and mobile guards track a mobile intruder with unknown maximum speed. We presented an algorithm to identify {\it candidate vertices} in a polygon that can serve as endpoints of diagonal guards or as a seat for vertex guards. We presented an activation strategy for the guards that is adaptive to the instantaneous speed of the intruder. Simulation results for two different environments illustrated the performance of the proposed techniques.
	
	In the future, we plan to extend our analysis to address the case of multiple intruders. Another direction of future research is to explore the effect of coordination among the intruders on the tracking performance. Other directions include addressing the tracking problem for guards with limited sensing and motion capabilities, for example, edge guards and line guards. 
	
	\bibliographystyle{spbasic}
	\bibliography{references3}
	
\end{document}